\documentclass[runningheads,a4paper]{llncs}

\usepackage{graphicx}
\usepackage[all]{xy}

\usepackage{cite}
\usepackage{amsmath,amssymb,amsfonts}
\usepackage{algorithmic}
\usepackage{graphicx}
\usepackage{textcomp}
\usepackage{xcolor}
\def\BibTeX{{\rm B\kern-.05em{\sc i\kern-.025em b}\kern-.08em
    T\kern-.1667em\lower.7ex\hbox{E}\kern-.125emX}}
    
\usepackage{stmaryrd,enumerate,bbm,tikz,calc}
\usetikzlibrary{arrows,arrows.meta,automata}
\usepackage{xspace}
\usepackage[justification=centering, labelfont=bf, textfont=sl]{caption}
\usepackage[bookmarks,unicode,colorlinks=true,citecolor=blue,linkcolor=blue,urlcolor=blue]{hyperref}

\newcommand\mksf[1]{\ensuremath{\mathsf{#1}}}
\newcommand\supp{\mathop{\mathsf{supp}}}
\renewcommand\sharp{\texttt{\upshape{\#}}}
\newcommand\Set{\ensuremath{\mathsf{Set}}}
\newcommand\Cat{\ensuremath{\mathsf{Cat}}}
\newcommand\set[2]{\{#1 \mid #2\}}
\newcommand\reals{{\mathbb R}}
\newcommand\naturals{{\mathbb N}}
\newcommand\len[1]{|#1|}
\newcommand\lam[2]{\lambda#1\kern1pt.\kern1pt#2}
\newcommand\Hom{\mathop{\mathsf{Hom}}}
\newcommand\bx[1]{\text{\rm\texttt[}#1\text{\rm\texttt]}}
\newcommand\seq[3]{#1_{#2},\ldots,#1_{#3}}
\newcommand\longline{\begin{flushleft}\mbox{}\hrulefill\mbox{}\end{flushleft}}
\newcommand\subs{\mathrel{\subseteq}}
\newcommand\CC{{\cal C}}
\newcommand\DD{{\cal D}}
\newcommand\EE{{\cal E}}
\newcommand\FF{{\cal F}}
\newcommand\OO{{\cal O}}
\newcommand\TT{{\cal T}}
\newcommand\RR{{\cal R}}
\newcommand\id{\mksf{id}}
\newcommand\Imp{\Rightarrow}
\renewcommand\phi\varphi

\begin{document}

\title{Multisets and Distributions\thanks{This paper is dedicated to Herman Geuvers on the occasion of his 60th birthday. Both authors spent time in Nijmegen and enjoyed many conversations with Herman, who has always been curious on the many aspects of programming semantics and categorical aspects. We hope he enjoys this simplified presentation of a complicated law, followed by a categorical generalization!}} 			
\author{Dexter Kozen \and Alexandra Silva}		
\institute{Cornell University}
\maketitle

\begin{abstract}
We give a lightweight alternative construction of Jacobs's distributive law for multisets and distributions that does not involve any combinatorics. We first give a distributive law for lists and distributions, then apply a general theorem on 2-categories that allows properties of lists to be transferred automatically to multisets. The theorem states that equations between 2-cells are preserved by epic 2-natural transformations. In our application, the appropriate epic 2-natural transformation is defined in terms of the Parikh map, familiar from formal language theory, that takes a list to its multiset of elements.
\end{abstract}

\section{Introduction}

A notoriously difficult challenge in the semantics of probabilistic programming is the combination of probability and nondeterminism. The chief difficulty is the lack of an appropriate distributive law between the powerset and distribution monads \cite{Beck69}. Many researchers have grappled with this issue over the years, proposing various workarounds \cite{Affeldt21,ChenSanders09,DahlqvistParlantSilva18,GoyPetrisan20,HartogdeVink99,KeimelPlotkin17,Mislove00,MisloveOuaknineWorrell04,Varacca03,VaraccaWinskel06,WangHoffmannReps19,Zwart20,ZwartMarsden22}. In contrast, if instead of the powerset one considers the multiset monad, it has been known for a while that indeed the multiset monad does distribute over the probability distributions monad. This observation appeared in e.g.~\cite{KeimelPlotkin17,DahlqvistParlantSilva18,Dash21}, but its origin is hard to trace as it seems to have been folklore knowledge for some time. Whereas the aforementioned references justified the existence of the law quite convincingly from an operational point of view, what was missing for a long time was an actual explicit description. 

Recently, an important breakthrough was provided by Jacobs \cite{Jacobs21}, who presented multiple explicit descriptions of the distributive law for finite multisets over finite distributions. His treatment is a combination of categorical and combinatorial reasoning, involving the commutativity of multinomial distributions with \emph{hypergeometric distributions}, discrete probability distributions that calculate the likelihood an event occurring $k$ times in $n$ trials when sampling from a small population without replacement. 

In this paper, we give a lighter-weight alternative to Jacobs's proof that does not involve any combinatorics. Our approach is based on the premise that \emph{lists} are more intuitive and notationally less cumbersome than multisets. There is an intuitively simple distributive law for lists over distributions: given a list of distributions, sample them independently to obtain a list of outcomes; the probability of that list of outcomes is the product of the probabilities of the components. We have thus turned a list of distributions into a distribution over lists. We show that this operational intuition is correct using purely equational reasoning, then apply a general theorem on 2-categories that maps properties of lists down to properties of multisets.

The general theorem, which may be of independent interest, states that \emph{equations between natural transformations are preserved by epic 2-natural transformations}. In our application, the appropriate epic 2-natural transformation is defined in terms of the \emph{Parikh map}, familiar from formal language theory, that takes a list to its multiset of elements. This theorem allows monad and distributive laws involving multisets to be obtained automatically from monad and distributive laws involving lists.

The contributions of this paper are twofold. On the one hand, we present a concrete construction of an important distributive law between multisets and distributions. Jacobs~\cite{Jacobs21} presents multiple constructions, all of which required more involved machinery. Here our argument is that a simpler presentation, without combinatorics, of the Beck distributive law is possible, though the alternative ones in \cite{Jacobs21} are valuable in different ways (as the author also argues). In particular, we do not claim that our approach is useful for calculating probabilities of events; indeed, such calculations require combinatorics, which we have eliminated in our approach. On the other hand, our general result provides a useful tool that may be applicable in other contexts. 

We start with some brief preliminaries in \S\ref{sec:prelims} introducing the multiset and distribution monads as submonads of a common monad. In \S\ref{sec:ld} we present the ingredients needed to define the desired distributive law between multisets and distributions and show how it can be obtained from a distributive law between lists and distributions via the Parikh map. In \S\ref{sec:gen} we present a general theorem on $2$-categories and show that many of the concrete results of the previous section are instances of the general theorem. 

\section{The distribution and multiset monads}
\label{sec:prelims} 

The distribution and multiset monads are both submonads of a common generalization. Consider the endofunctor
\begin{align*}
& F:\Set\to\Set && FX = \set{f\in\reals_+^X}{\len{\supp f}<\infty},
\end{align*}
the set of nonnegative real-valued functions $f:X\to\reals$ with finite support, and for $h:X\to Y$,
\begin{align*}
& Fh:FX\to FY && Fh(f) = \lam{b\in Y}{\sum_{h(a)=b} f(a)}.
\end{align*}
Note that $F$ is covariant, unlike the usual hom-functor $\Hom(-,\reals_+)$.
For $f\in FX$, let $\len f=\sum_{a\in X}f(a)$.

The functor $F$ carries a monad structure $(F,\mu,\eta)$ with
$\eta:I\to F$ and $\mu:F^2\to F$, where for $a\in X$ and $f\in F^2X$,
\begin{align*}
\eta_X(a) &= \lam{b\in X}{[b=a]\footnotemark}\\
\mu_X(f) &= \lam{a\in X}{\sum_{g\in FX}{f(g)\cdot g(a)}}.
\end{align*}
\footnotetext{\small\emph{Iverson bracket}: $[\phi] = 1$ if $\phi$ is true, 0 if false}%
Restricting the range of functions to $\naturals$, we get a submonad $(M,\mu,\eta)$, where
\begin{align*}
MX &= \set{m\in FX}{\forall a\in X\ m(a)\in\naturals}.
\end{align*}
This is the \emph{multiset monad}. The value $m(a)$ is the multiplicity of $a\in X$ in the multiset $m$.
On the other hand, restricting to convex functions, we get a submonad $(D,\mu,\eta)$, where
\begin{align*}
DX &= \set{d\in FX}{\sum_{a\in X}d(a)=1}.
\end{align*}
This is the \emph{distribution monad}. The value $d(a)$ is the probability of $a\in X$ under the distribution $d$.

Clearly $\mu^M:M^2\to M$, as $\naturals$ is closed under the semiring operations, and $\mu^D:D^2\to D$, as a convex function of convex functions flattens to a convex function.

We will sometimes annotate $\eta$ and $\mu$ with superscripts $M$ and $D$ to emphasize that we are interpreting them in the respective submonads, but the definitions are the same. In more conventional notation, $\eta^D_X(a)=\delta_a$, the Dirac (point-mass) distribution on $a\in X$, and distributions can be expressed as weighted formal sums with real coefficients as in \cite{Jacobs21}.

The Eilenberg-Moore algebras for $F$ are the $\reals_+$-semimodules, 
although we do not need to know this for our study.

In the following sections, we will give a construction of a distributive law $\otimes^M:MF\to FM$, which contains the desired distributive law $MD\to DM$ as a special case.

\section{Lists and distributions}
\label{sec:ld}

In this section we set the stage by giving an intuitive account of our approach. It will be evident that it is possible to obtain the desired distributive law in a piecemeal fashion using only arguments in this section; but in \S\ref{sec:gen} we will show how to encapsulate these arguments in a general theorem to obtain them all at once. Nevertheless, the development of this section is essential motivation for understanding the more abstract version of \S\ref{sec:gen}.

\subsection{A distributive law $\otimes^L:LF\to FL$}

Let $(L,\mu^L,\eta^L)$ be the list monad. Define $\otimes^L_X:LFX\to FLX$, where
\begin{align*}
& \otimes^L_X(\bx{\seq f1n})\\
&= \lam{\bx{\seq a1m}\in LX}{\begin{cases}
\prod_{i=1}^n f_i(a_i) & \text{if $m=n$}\\
0 & \text{otherwise.\footnotemark[2]}
\end{cases}}
\end{align*}
\footnotetext[2]{\small$\bx{\seq a1n}$ denotes a list of $n$ elements, not to be confused with the Iverson bracket $[\phi]$. We trust the reader will be able to distinguish them by the difference in fonts and by context.}%
We abbreviate this by writing
\begin{align*}
\otimes^L_X(\bx{\seq f1n}) = 
\lam{\bx{\seq a1n}\in X^n}{\prod_{i=1}^n f_i(a_i)}
\end{align*}
with the convention that the right-hand side, while defined on all of $LX$, vanishes outside $X^n$.
The support $\supp\otimes^L_X(\bx{\seq f1n})$ is the cartesian product $\prod_{i=1}^n\supp f_i$.

\begin{theorem}
\label{thm:LF}
$\otimes^L:LF\to FL$ is a distributive law for $(L,\mu^L,\eta^L)$ over $(F,\mu,\eta)$. That is, it is a natural transformation and satisfies the axioms of Fig.~\ref{fig:distribL}.
\begin{figure}[t]
\begin{align*}
&
\begin{tikzpicture}[->, >=stealth', auto]
\scriptsize
\node (NW) {$LLF$};
\node (N) [right of=NW, node distance=18mm] {$LFL$};
\node (NE) [right of=N, node distance=18mm] {$FLL$};
\node (SW) [below of=NW, node distance=12mm] {$LF$};
\node (SE) [below of=NE, node distance=12mm] {$FL$};
\path (NW) edge node {$L\otimes^L$} (N);
\path (N) edge node {$\otimes^L L$} (NE);
\path (NW) edge node[swap] {$\mu^LF$} (SW);
\path (NE) edge node {$F\mu^L$} (SE);
\path (SW) edge node[swap] {$\otimes^L$} (SE);
\end{tikzpicture}
&&
\begin{tikzpicture}[->, >=stealth', auto, node distance=48mm]
\scriptsize
\node (N) {$L$};
\node (SW) [below left of=N, node distance=16mm, xshift=3mm] {$LF$};
\node (SE) [below right of=N, node distance=16mm, xshift=-3mm] {$FL$};
\path (N) edge node[swap] {$L\eta^F$} (SW);
\path (N) edge node {$\eta^FL$} (SE);
\path (SW) edge node[swap] {$\otimes^L$} (SE);
\end{tikzpicture}\\
&
\begin{tikzpicture}[->, >=stealth', auto]
\scriptsize
\node (NW) {$LFF$};
\node (N) [right of=NW, node distance=18mm] {$FLF$};
\node (NE) [right of=N, node distance=18mm] {$FFL$};
\node (SW) [below of=NW, node distance=12mm] {$LF$};
\node (SE) [below of=NE, node distance=12mm] {$FL$};
\path (NW) edge node {$\otimes^LF$} (N);
\path (N) edge node {$F\otimes^L$} (NE);
\path (NW) edge node[swap] {$L\mu^F$} (SW);
\path (NE) edge node {$\mu^FL$} (SE);
\path (SW) edge node[swap] {$\otimes^L$} (SE);
\end{tikzpicture}
&&
\begin{tikzpicture}[->, >=stealth', auto, node distance=48mm]
\scriptsize
\node (N) {$F$};
\node (SW) [below left of=N, node distance=16mm, xshift=3mm] {$LF$};
\node (SE) [below right of=N, node distance=16mm, xshift=-3mm] {$FL$};
\path (N) edge node[swap] {$\eta^LF$} (SW);
\path (N) edge node {$F\eta^L$} (SE);
\path (SW) edge node[swap] {$\otimes^L$} (SE);
\end{tikzpicture}
\end{align*}
\caption{Axioms for the distributive law $\otimes^L:LF\to FL$}
\label{fig:distribL}
\longline
\end{figure}
\end{theorem}
\begin{proof}
The proof uses only elementary equational reasoning and is given in its entirety in \S\ref{sec:distribL}.
\end{proof}

Applied to a list of $n$ distributions over $X$, $\otimes^L_X$ gives the joint distribution on the cartesian product $X^n$ obtained by sampling the component distributions independently.
\begin{align*}
\otimes^L_X(\bx{\seq d1n})(\bx{\seq a1n}) = \prod_{i=1}^n d_i(a_i).
\end{align*}
As $\otimes^L$ produces a distribution on $LX$ when applied to a list of distributions on $X$, it specializes to $\otimes^L:LD\to DL$.

We would like to have a similar distributive law $\otimes^M:MF\to FM$, which will also specialize to $\otimes^M:MD\to DM$ when applied to a multiset of distributions. We will show how to derive $\otimes^M$ from the distributive law $\otimes^L:LF\to FL$ and a natural transformation $\sharp:L\to M$, the \emph{Parikh map} familiar from formal language theory.

\subsection{The Parikh map $\sharp:L\to M$}

Consider the monoid $(LX,\cdot,\bx{})$, where $(\cdot)$ denotes list concatenation (often elided) and $\bx{}$ is the empty list. This is the free monoid on generators $X$. Consider also the monoid $(MX,+,0)$, where $+$ is pointwise sum of real-valued functions and $0$ is the constant-zero function, representing multiset union and the empty multiset, respectively. This is the free commutative monoid on generators $X$. The two monoids are connected by the \emph{Parikh map} $\sharp_X:LX\to MX$, where $\sharp_X(x)(a)$ gives the number of occurrences of $a\in X$ in the list $x\in LX$. This is a monoid homomorphism from the free monoid on generators $X$ to its commutative image. Note that $\sharp_X$ is an epimorphism, as every multiset is the image of a list. Jacobs \cite{Jacobs21} refers to the map $\sharp_X$ as \emph{accumulation} (\emph{acc}).

More generally, we will show below that $\sharp$ is a morphism of monads
\begin{align*}
\sharp:(L,\mu^L,\eta^L)\to(M,\mu^M,\eta^M)
\end{align*}
(Theorem \ref{thm:LLMM}). 
This means that $\sharp:L\to M$ is a natural transformation that commutes with $\mu$ and $\eta$; in other words, it is a monoid homomorphism, regarding monads as monoids over the category of endofunctors.

Because $\sharp_X$ is a monoid homomorphism, we have
\begin{align*}
\sharp_X(\bx{}) &= 0 & \sharp_X(xy) &= \sharp_X(x) + \sharp_X(y).
\end{align*}
On generators $\bx a$, 
\begin{align*}
\sharp_X(\bx{a}) &= \lam{b\in X}{[b=a]} = \eta_X(a).
\end{align*}
Using these facts, we can show that $\sharp:L\to M$ is natural transformation. For $f:X\to Y$,
\begin{align*}
& Mf(\sharp_X(\bx{\seq a1n}))\\
&= Mf(\sum_{i=1}^n \sharp_X(\bx{a_i}))
= Mf(\sum_{i=1}^n \lam{a\in X}{[a=a_i]})\\
&= \lam{b\in Y}{\sum_{f(a)=b} \sum_{i=1}^n [a=a_i]}
= \lam{b\in Y}{\sum_{i=1}^n [b=f(a_i)]}\\
&= \sum_{i=1}^n \sharp_Y(\bx{f(a_i)})
= \sharp_Y(\bx{f(a_1),\ldots,f(a_n)})\\
&= \sharp_Y(Lf(\bx{\seq a1n})),
\end{align*}
thus $Mf\circ\sharp_X = \sharp_Y\circ Lf$.

\subsection{A distributive law $\otimes^M:MF\to FM$}

We first define the natural transformation $\otimes^M:MF\to FM$, then show that it satisfies the requisite properties of a distributive law. For any $x\in LX$, define
\begin{align*}
S(x)
&= \set{y\in LX}{\sharp_X(y)=\sharp_X(x)}
= \sharp_X^{-1}(\sharp_X(x)).
\end{align*}
Then $y\in S(x)$ iff $y$ and $x$ have the same length, say $n$, and $y$ is a permutation of $x$; that is,
\begin{align*}
S(\bx{\seq b1n}) &= \set{\bx{\seq b{\sigma(1)}{\sigma(n)}}}{\text{$\sigma$ is a permutation on $\{1,\ldots,n\}$}}.
\end{align*}

\begin{lemma}
\label{lem:pi}
Let $S=S(\bx{\seq b1n})$. For any fixed permutation $\pi$ on $\{1,\ldots,n\}$, the map
\begin{align}
\bx{\seq a1n} \mapsto \bx{\seq a{\pi(1)}{\pi(n)}}\label{eq:pi}
\end{align}
is a bijection $S\to S$.
\end{lemma}
\begin{proof}
Since $S$ is closed under permutations,
\begin{align*}
\set{\bx{\seq a{\pi(1)}{\pi(n)}}}{\bx{\seq a1n}\in S} \subs S.
\end{align*}
But the map \eqref{eq:pi} is injective, so
\begin{align*}
\len{\set{\bx{\seq a{\pi(1)}{\pi(n)}}}{\bx{\seq a1n}\in S}} \ge \len S,
\end{align*}
therefore equality holds, so \eqref{eq:pi} must be a bijection on $S$.
\qed
\end{proof}

The desired natural transformation $\otimes^M:MF\to FM$ is defined in the following lemma. This is essentially the same as Jacob's third characterization \cite[Equation (5)]{Jacobs21}.
\begin{lemma}
\label{lem:otimesdef}
There is a unique natural transformation $\otimes^M:MF\to FM$ such that the following diagram commutes:
\begin{align}
\begin{array}c
\begin{tikzpicture}[->, >=stealth', auto, node distance=20mm]
\small
\node (NW) {$LF$};
\node (NE) [right of=NW] {$FL$};
\node (SW) [below of=NW, node distance=12mm] {$MF$};
\node (SE) [below of=NE, node distance=12mm] {$FM$};
\path (NW) edge node {$\otimes^L$} (NE);
\path (NW) edge node[swap] {$\sharp F$} (SW);
\path (NE) edge node {$F\sharp$} (SE);
\path (SW) edge node[swap] {$\otimes^M$} (SE);
\end{tikzpicture}
\end{array}
\label{eq:otimesdef}
\end{align}
\end{lemma}
\begin{proof}
Let
\begin{align*}
X^{(n)}=\set{m\in MX}{\sum_{a\in X}m(a)=n},
\end{align*}
the multisets over $X$ of size $n$. For a given $\bx{\seq f1n}\in LFX$, define
\begin{align*}
& \otimes^M_X(\sharp_{FX}(\bx{\seq f1n}))\\
&= \lam{y\in X^{(n)}}{\sum_{\bx{\seq a1n}\in\sharp_X^{-1}(y)} \prod_{i=1}^n f_i(a_i)}.
\end{align*}
Again by convention, it is assumed that the function on the right-hand side is defined on all of $MX$, but vanishes on $MX\setminus X^{(n)}$.
This defines $\otimes^M_X$ on the multiset $\sharp_{FX}(\bx{\seq f1n})$; we must show that the definition is independent of the choice of $\bx{\seq f1n}$.

For any $\bx{\seq b1n}\in LX$, we have $\sharp_X(\bx{\seq b1n})\in MX$, and
\begin{align*}
& \otimes^M_X(\sharp_{FX}(\bx{\seq f1n}))(\sharp_X(\bx{\seq b1n}))\\
&= {\sum_{\bx{\seq a1n}\in\sharp_X^{-1}(\sharp_X(\bx{\seq b1n}))} \prod_{i=1}^n f_i(a_i)}\\
&= {\sum_{\bx{\seq a1n}\in S(\bx{\seq b1n})} \prod_{i=1}^n f_i(a_i)}.
\end{align*}
Then for any permutation $\sigma:\{1,\ldots,n\}\to\{1,\ldots,n\}$,
\begin{align}
& \otimes^M_X(\sharp_{FX}(\bx{\seq f{\sigma(1)}{\sigma(n)}}))(\sharp_X(\bx{\seq b1n}))\nonumber\\
&= {\sum_{\bx{\seq a1n}\in S(\bx{\seq b1n})} \prod_{i=1}^n f_{\sigma(i)}(a_i)}\nonumber\\
&= {\sum_{\bx{\seq a1n}\in S(\bx{\seq b1n})} \prod_{i=1}^n f_i(a_{\sigma^{-1}(i)})}\label{eq:perms1}\\
&= {\sum_{\bx{\seq a{\sigma^{-1}(1)}{\sigma^{-1}(n)}}\in S(\bx{\seq b1n})} \prod_{i=1}^n f_i(a_{\sigma^{-1}(i)})}\label{eq:perms2}\\
&= {\sum_{\bx{\seq a1n}\in S(\bx{\seq b1n})} \prod_{i=1}^n f_i(a_i)}.\label{eq:perms3}
\end{align}
The inference \eqref{eq:perms1} is from the commutativity of multiplication,
\eqref{eq:perms2} is from Lemma \ref{lem:pi}, and \eqref{eq:perms3} comes from reindexing.
It follows that $\otimes^M_X(\sharp_{FX}(\bx{\seq f1n}))$ and $\otimes^M_X(\sharp_{FX}(\bx{\seq f{\sigma(1)}{\sigma(n)}}))$ agree on all elements of the form $\sharp_X(\bx{\seq b1n})\in MX$, and both take value 0 on elements of the form $\sharp_X(\bx{\seq b1m})$ for $m\ne n$. But this is all of $MX$, as $\sharp_X$ is surjective.

The equation \eqref{eq:otimesdef} now follows from the definitions of $F\sharp_X$ and $\otimes^L_X$:
\begin{align*}
& \otimes^M_X(\sharp_{FX}(\bx{\seq f1n}))\\
&= \lam{y\in X^{(n)}}{\sum_{\bx{\seq a1n}\in\sharp_X^{-1}(y)} \prod_{i=1}^n f_i(a_i)}\\
&= \lam{y\in X^{(n)}}{\sum_{\sharp_X(\bx{\seq a1n})=y} \prod_{i=1}^n f_i(a_i)}\\
&= F\sharp_X(\otimes^L_X(\bx{\seq f1n})).
\end{align*}

It remains to show the naturality of $\otimes^M$. That is, for $f:X\to Y$,
\begin{align}
FMf\circ\otimes^M_X = \otimes^M_Y\circ MFf.\label{eq:natotimesM}
\end{align}
We will take this opportunity to illustrate the general observation that allows us to map properties interest from $L$ down to $M$ via $\sharp$. Suppose we know that $\otimes^L$ is a natural transformation.
\begin{figure}[t]
\begin{align*}
\begin{array}c
\begin{tikzpicture}[->, >=stealth', auto, node distance=28mm]
\scriptsize
\node (NW) {$LFX$};
\node (NE) [right of=NW] {$FLX$};
\node (SW) [below right of=NW, node distance=14mm, xshift=6mm] {$LFY$};
\node (SE) [below right of=NE, node distance=14mm, xshift=10mm] {$FLY$};
\path (NW) edge node {$\otimes^L_X$} (NE);
\path (NW) edge node[swap,xshift=1mm] {$LFf$} (SW);
\path (NE) edge node[xshift=-1mm] {$FLf$} (SE);
\path (SW) edge node[pos=.6] {$\otimes^L_Y$} (SE);
\node (uNW) [below of=NW, node distance=20mm] {$MFX$};
\node (uNE) [below of=NE, node distance=20mm] {$FMX$};
\node (uSW) [below of=SW, node distance=20mm] {$MFY$};
\node (uSE) [below of=SE, node distance=20mm] {$FMY$};
\path (uNW) edge node[swap, pos=.8] {$\otimes^M_X$} (uNE);
\path (uNW) edge node[swap,xshift=1mm] {$MFf$} (uSW);
\path (uNE) edge node[xshift=-1mm] {$FMf$} (uSE);
\path (uSW) edge node[swap] {$\otimes^M_Y$} (uSE);
\path (NW) edge node[swap] {$\sharp_{FX}$} (uNW);
\path (NE) edge node[pos=.75] {$F\sharp_X$} (uNE);
\path (SW) edge node[swap,pos=.25] {$\sharp_{FY}$} (uSW);
\path (SE) edge node {$F\sharp_Y$} (uSE);
\end{tikzpicture}
\end{array}
\end{align*}
\caption{Naturality of $\otimes^M$ from $\otimes^L$}
\label{fig:nattrans}
\longline
\end{figure}
In Fig.~\ref{fig:nattrans}, the commutativity of the top face asserts this fact, and we wish to derive the same for the bottom face. The front and back vertical faces commute by \eqref{eq:otimesdef}. The left and right vertical faces commute by the naturality of $\sharp$. Using these facts, we can show the following paths in Fig.~\ref{fig:nattrans} are equivalent:
\begin{align}
& LFX \to MFX \to FMX \to FMY\nonumber\\
&=\ LFX \to FLX \to FLY \to FMY \label{eq:paths1}\\
&=\ LFX \to LFY \to FLY \to FMY \label{eq:paths2}\\
&=\ LFX \to MFX \to MFY \to FMY \label{eq:paths3}
\end{align}
The inference \eqref{eq:paths1} is from the back and right vertical faces,
\eqref{eq:paths2} is from the top face, and 
\eqref{eq:paths3} is from the left and front vertical faces.
Thus
\begin{align*}
FMf\circ\otimes_X^M\circ\sharp_{FX} = \otimes_Y^M\circ MFf\circ\sharp_{FX}.
\end{align*}
But since $\sharp_{FX}:LFX\to MFX$ is an epimorphism, we can conclude that
\begin{align*}
FMf\circ\otimes_X^M = \otimes_Y^M\circ MFf.
\tag*{\qed}
\end{align*}
\end{proof}

We can apply the same technique to show that $\otimes^M$ satisfies the axioms of distributive laws. These are the same as those of Fig.~\ref{fig:distribL} with $M$ in place of $L$.
Assuming these diagrams hold for $L$ (which we still have to prove---see \S\ref{sec:distribL}), we can derive them for $M$. For example,
\begin{figure}[t]
\begin{align*}
\begin{tikzpicture}[->, >=stealth', auto, node distance=24mm]
\scriptsize
\node (NW) {$LLF$};
\node (N) [right of=NW] {$LFL$};
\node (NE) [right of=N] {$FLL$};
\node (SW) [below right of=NW, node distance=10mm, xshift=4mm] {$LF$};
\node (SE) [below right of=NE, node distance=10mm, xshift=10mm] {$FL$};
\path (NW) edge node {$L\otimes^L$} (N);
\path (N) edge node {$\otimes^L L$} (NE);
\path (NW) edge node[swap,xshift=2mm] {$\mu^LF$} (SW);
\path (NE) edge node {$F\mu^L$} (SE);
\path (SW) edge node {$\otimes^L$} (SE);
\node (uNW) [below of=NW, node distance=20mm] {$MMF$};
\node (uN) [below of=N, node distance=20mm] {$MFM$};
\node (uNE) [below of=NE, node distance=20mm] {$FMM$};
\node (uSW) [below of=SW, node distance=20mm] {$MF$};
\node (uSE) [below of=SE, node distance=20mm] {$FM$};
\path (uNW) edge node[swap, pos=.8] {$M\otimes^M$} (uN);
\path (uN) edge node[swap] {$\otimes^M M$} (uNE);
\path (uNW) edge node[swap,xshift=1mm] {$\mu^MF$} (uSW);
\path (uNE) edge node[xshift=-1mm] {$F\mu^M$} (uSE);
\path (uSW) edge node[swap] {$\otimes^M$} (uSE);
\path (NW) edge node[swap] {$\sharp\sharp F$} (uNW);
\path (N) edge node[pos=.75] {$\sharp F\sharp$} (uN);
\path (NE) edge node[pos=.75] {$F\sharp\sharp$} (uNE);
\path (SW) edge node[swap, pos=.4] {$\sharp F$} (uSW);
\path (SE) edge node {$F\sharp$} (uSE);
\end{tikzpicture}
\end{align*}
\caption{An axiom of distributive laws}
\label{fig:distrib}
\longline
\end{figure}
to transfer the upper left property of Fig.~\ref{fig:distribL} from $L$ to $M$, we use the diagram of Fig.~\ref{fig:distrib}, where
\begin{align*}
& \sharp\sharp = M\sharp\circ\sharp L = \sharp M\circ L\sharp\\
& \sharp F\sharp = MF\sharp\circ\sharp FL = \sharp FM\circ LF\sharp.
\end{align*}
The top face of Fig.~\ref{fig:distrib} asserts the property of interest for $L$, and we wish to derive the same for the bottom face. We obtain an equivalence of paths
\begin{align*}
& LLF \to MMF \to MFM \to FMM \to FM\\
&=\ LLF \to MMF \to MF \to FM
\end{align*}
using the commutativity of the top and vertical side faces, then use the fact that $\sharp\sharp F:LLF\to MMF$ is an epimorphism to conclude that the bottom face commutes. The commutativity of the top face is proved below in \S\ref{sec:distribL}; the side faces are from \eqref{eq:otimesdef}, the naturality of $\sharp$, and the monad laws.

The verification of the other axioms proceeds similarly. In fact, in \S\ref{sec:gen} we will prove a general theorem that allows us to derive these all at once.

\begin{theorem}
\label{thm:distribM}
$\otimes^M:MF\to FM$ is a distributive law. Restricted to multisets of distributions, it specializes to $\otimes^M:MD\to DM$.
\end{theorem}
\begin{proof}
The verification of the four distributive law axioms can be done individually along the lines of the argument above, or all at once using Corollary \ref{cor:nattranssharp} below.

For the last statement, we observe that for distributions $\seq d1n$,
\begin{align*}
& \sum_{m\in MX}\otimes^M_X(\sharp_{FX}(\bx{\seq d1n}))(m)\\
&= \sum_{m\in MX}\sum_{\sharp_X(\bx{\seq a1n})=m} \prod_{i=1}^n d_i(a_i)\\
&= \sum_{\bx{\seq a1n}\in LX}\prod_{i=1}^n d_i(a_i)\ =\ \prod_{i=1}^n\sum_{a\in X} d_i(a)\ =\ 1.
\tag*{\qed}
\end{align*}
\end{proof}
\ 
\section{A general result}
\label{sec:gen}


In this section we show how results of the previous section can be obtained as a consequence of a general theorem on 2-categories to the effect that \emph{equations between 2-cells are preserved under epic 2-natural transformations}. In our application, this theorem allows equations involving $L$ and $F$ to be mapped down to equations involving $M$ and $F$ via a 2-natural transformation defined in terms of $\sharp$.

Recall that in a 2-category, composition of 1-cells and composition of 2-cells are called \emph{horizontal} and \emph{vertical composition}, respectively. We will write $\circ$ for vertical composition and denote horizontal composition by juxtaposition. Horizontal composition also acts on 2-cells and satisfies the property:
If $F,F',F'':\CC\to\DD$ and $G,G',G'':\DD\to\EE$ are 1-cells and $\phi:F\to F'$, $\phi':F'\to F''$, $\psi:G\to G'$, and $\psi':G'\to G''$ are 2-cells, then
\begin{align}
(\phi'\circ\phi)(\psi'\circ\psi) &= (\phi'\psi')\circ(\phi\psi) : FG\to F''G''.\label{eq:2cell}
\end{align}
\begin{align*}
\begin{tikzpicture}[->, >=stealth', auto, bend angle=60, node distance=25mm]
\small
\node (CC) {$\CC$};
\node (DD) [right of=CC] {$\DD$};
\node (EE) [right of=DD] {$\EE$};
\path (CC) edge node[pos=.2] {$F'$} (DD);
\path (CC) edge[bend left] node {$F$} (DD);
\path (CC) edge[bend right] node[swap] {$F''$} (DD);
\path (DD) edge node[pos=.2] {$G'$} (EE);
\path (DD) edge[bend left] node {$G$} (EE);
\path (DD) edge[bend right] node[swap] {$G''$} (EE);
\draw [double distance=2pt, arrows = {-Implies[]}] (1.25,.7) -- (1.25,.1);
\node at (1.6,.4) {$\phi$};
\draw [double distance=2pt, arrows = {-Implies[]}] (1.25,-.1) -- (1.25,-.7);
\node at (1.6,-.3) {$\phi'$};
\draw [double distance=2pt, arrows = {-Implies[]}] (3.75,.7) -- (3.75,.1);
\node at (4.1,.4) {$\psi$};
\draw [double distance=2pt, arrows = {-Implies[]}] (3.75,-.1) -- (3.75,-.7);
\node at (4.1,-.3) {$\psi'$};
\node at (6,0) {$=$};
\node (CCp) at (7,0) {$\CC$};
\node (EEp) at (10,0) {$\EE$};
\path (CCp) edge[bend angle=50, bend left] node {$FG$} (EEp);
\path (CCp) edge node[pos=.2] {$F'G'$} (EEp);
\path (CCp) edge[bend angle=50, bend right] node[swap] {$F''G''$} (EEp);
\draw [double distance=2pt, arrows = {-Implies[]}] (8.5,.7) -- (8.5,.1);
\node at (9,.4) {$\phi\psi$};
\draw [double distance=2pt, arrows = {-Implies[]}] (8.5,-.1) -- (8.5,-.7);
\node at (9,-.3) {$\phi'\psi'$};
\end{tikzpicture}
\end{align*}

Categories, functors, and natural transformations form a 2-category \Cat, in which the 0-cells are categories, the 1-cells are functors, and the 2-cells are natural transformations. In \Cat, horizontal composition on natural transformations is defined by: for $\phi:F\to G$ and $\psi:H\to K$,
\begin{align}
\phi\psi:FH\to GK &= \phi K\circ F\psi = G\psi\circ \phi H\label{eq:horizC}
\end{align}
\begin{align*}
\begin{tikzpicture}[->, >=stealth', auto, bend angle=60, node distance=16mm]
\small
\node (FH) {$FH$};
\node (FK) [below left of=FH] {$FK$};
\node (GH) [below right of=FH] {$GH$};
\node (GK) [below right of=FK] {$GK$};
\path (FH) edge[double distance=2pt, arrows = {-Implies[]}] node {$\phi\psi$} (GK);
\path (FH) edge[double distance=2pt, arrows = {-Implies[]}] node[swap] {$F\psi$} (FK);
\path (FH) edge[double distance=2pt, arrows = {-Implies[]}] node {$\phi H$} (GH);
\path (FK) edge[double distance=2pt, arrows = {-Implies[]}] node[swap] {$\phi K$} (GK);
\path (GH) edge[double distance=2pt, arrows = {-Implies[]}] node {$G\psi$} (GK);
\end{tikzpicture}
\end{align*}
and conversely, $\phi H=\phi\,\id_H$ and $F\psi=\id_F\psi$.

A \emph{2-natural transformation} between 2-functors is like a natural transformation between functors, but one level up. Formally, let $\FF$ and $\CC$ be 2-categories and let $\Phi,\Psi:\FF\to\CC$ be 2-functors such that $\Phi O=\Psi O$ for each 0-cell $O$ of $\FF$. A \emph{2-natural transformation} $\phi:\Phi\to\Psi$ is a collection of 2-cells $\phi_S:\Phi S\to\Psi S$ indexed by 1-cells $S$ such that for each 2-cell $\sigma:S\to T$, the following diagram commutes:
\begin{align}
\begin{array}c
\begin{tikzpicture}[->, >=stealth', auto, node distance=20mm]
\small
\node (NW) {$\Phi S$};
\node (NE) [right of=NW] {$\Phi T$};
\node (SW) [below of=NW, node distance=12mm] {$\Psi S$};
\node (SE) [below of=NE, node distance=12mm] {$\Psi T$};
\path (NW) edge node {$\Phi\sigma$} (NE);
\path (NW) edge node[swap] {$\phi_S$} (SW);
\path (NE) edge node {$\phi_T$} (SE);
\path (SW) edge node[swap] {$\Psi\sigma$} (SE);
\end{tikzpicture}
\end{array}
\label{dia:2nat}
\end{align}
The restriction that $\Phi O=\Psi O$ for each 0-cell $O$ of $\FF$ allows $\phi_S$ to be a 2-cell. A 2-natural transformation is \emph{epic} if all components are epimorphisms.

Now let $\FF$ be the free 2-category on finite sets of generators $\OO$, $\TT$, and $\RR$ for the 0-, 1-, and 2-cells, respectively. Let $\CC$ be another 2-category, and let $\Phi,\Psi:\FF\to\CC$ be 2-functors such that $\Phi O=\Psi O$ for each $O\in\OO$.

Suppose that to each 1-cell generator $S\in\TT$ there is associated a 2-cell $(\sharp)_S:\Phi S\to\Psi S$ of $\CC$.
If $S = S_1\cdots S_n$ is a 1-cell of $\FF$, define $(\sharp)_S$ to be the horizontal composition
\begin{align*}
(\sharp)_S &= (\sharp)_{S_1}\cdots(\sharp)_{S_n}:\Phi S\to\Psi S.
\end{align*}
In particular, for $n=0$, for $I_O:O\to O$ an identity 1-cell of $\FF$, $(\sharp)_{I_O}:\Phi I_O\to\Psi I_O$ is the identity 1-cell on $\Phi O=\Psi O$.

\begin{theorem}
\label{thm:nattranssharp}
If the diagram
\begin{align}
\begin{array}c
\begin{tikzpicture}[->, >=stealth', auto, node distance=20mm]
\small
\node (NW) {$\Phi S$};
\node (NE) [right of=NW] {$\Phi T$};
\node (SW) [below of=NW, node distance=12mm] {$\Psi S$};
\node (SE) [below of=NE, node distance=12mm] {$\Psi T$};
\path (NW) edge node {$\Phi\sigma$} (NE);
\path (NW) edge node[swap] {$(\sharp)_S$} (SW);
\path (NE) edge node {$(\sharp)_T$} (SE);
\path (SW) edge node[swap] {$\Psi\sigma$} (SE);
\end{tikzpicture}
\end{array}
\label{dia:nattranssharp}
\end{align}
commutes for each 2-cell generator $\sigma:S\to T$ in $\RR$, then it commutes for all 2-cells $\sigma$ generated by $\RR$. Consequently, $(\sharp)$ is a 2-natural transformation $\Phi\to\Psi$ with components $(\sharp)_S$.
\end{theorem}
\begin{proof}
For identities,
\begin{align*}
(\sharp)_S\circ\Phi(\id_S)
&= (\sharp)_S\circ\id_{\Phi S} && \text{since $\Phi$ is a 2-functor}\\
&= \id_{\Psi S}\circ(\sharp)_S\\
&= \Psi(\id_S)\circ(\sharp)_S && \text{since $\Psi$ is a 2-functor.}
\end{align*}
For vertical composition $\tau\circ\sigma:S\to U$ with $\sigma:S\to T$ and $\tau:T\to U$, if \eqref{dia:nattranssharp} commutes for $\sigma$ and $\tau$, that is,
\begin{align}
& (\sharp)_T\circ\Phi\sigma=\Psi\sigma\circ(\sharp)_S && (\sharp)_U\circ\Phi\tau=\Psi\tau\circ(\sharp)_T,\label{eq:vert}
\end{align}
then
\begin{align*}
& (\sharp)_U\circ\Phi(\tau\circ\sigma)\\
&= (\sharp)_U\circ\Phi\tau\circ\Phi\sigma && \text{since $\Phi$ is a 2-functor}\\
&= \Psi\tau\circ(\sharp)_T\circ\Phi\sigma && \text{by \eqref{eq:vert}, right-hand equation}\\
&= \Psi\tau\circ\Psi\sigma\circ(\sharp)_S && \text{by \eqref{eq:vert}, left-hand equation}\\
&= \Psi(\tau\circ\sigma)\circ(\sharp)_S && \text{since $\Psi$ is a 2-functor.}
\end{align*}
For horizontal composition $\sigma\tau:SU\to TV$ with $\sigma:S\to T$ and $\tau:U\to V$, if \eqref{dia:nattranssharp} commutes for $\sigma$ and $\tau$, that is,
\begin{align}
& (\sharp)_T\circ\Phi\sigma=\Psi\sigma\circ(\sharp)_S && (\sharp)_V\circ\Phi\tau=\Psi\tau\circ(\sharp)_U,\label{eq:horiz}
\end{align}
then
\begin{align*}
& (\sharp)_{TV}\circ\Phi(\sigma\tau)\\
&= ((\sharp)_{T}(\sharp)_{V})\circ((\Phi\sigma)(\Phi\tau)) && \text{since $\Phi$ is a 2-functor}\\
&= ((\sharp)_{T}\circ\Phi\sigma)((\sharp)_{V}\circ\Phi\tau) && \text{by \eqref{eq:2cell}}\\
&= (\Psi\sigma\circ(\sharp)_{S})(\Psi\tau\circ(\sharp)_{U}) && \text{by \eqref{eq:horiz}}\\
&= ((\Psi\sigma)(\Psi\tau))\circ((\sharp)_{S}(\sharp)_{U}) && \text{by \eqref{eq:2cell}}\\
&= \Psi(\sigma\tau)\circ(\sharp)_{SU} && \text{since $\Psi$ is a 2-functor.}
\tag*{\qed}
\end{align*}
\end{proof}

\begin{corollary}
\label{cor:nattranssharp}
Suppose the conditions of Theorem \ref{thm:nattranssharp} hold. If in addition $(\sharp)_S$ is an epimorphism for every 1-cell $S$ of $\FF$, then every 2-cell equation that holds under the interpretation $\Phi$ also holds under the interpretation $\Psi$. That is, for all 2-cells $\sigma,\tau:S\to T$ of $\FF$, if $\Phi\sigma=\Phi\tau$, then $\Psi\sigma=\Psi\tau$.
\end{corollary}
\begin{proof}
We have
\begin{align*}
\Psi\sigma\circ(\sharp)_S
&= (\sharp)_T\circ\Phi\sigma && \text{by Theorem \ref{thm:nattranssharp}}\\
&= (\sharp)_T\circ\Phi\tau && \text{since $\Phi\sigma=\Phi\tau$}\\
&= \Psi\tau\circ(\sharp)_S && \text{by Theorem \ref{thm:nattranssharp}.}
\end{align*}
Since $(\sharp)_S$ is an epimorphism, $\Psi\sigma=\Psi\tau$.
\end{proof}

In our application, $\CC$ is \Cat, the 2-category of categories, functors, and natural transformations, and
\begin{itemize}
\item
$\OO = \{O\}$,
\item
$\TT = \{P:O\to O,Q:O\to O\}$,
\item
$\RR = \{\mu:PP\to P,\eta:I\to P,\\
\phantom{ddddd}\mu':QQ\to Q,\eta':I\to Q,\otimes:PQ\to QP\}$,
\item
$\Phi = \{O\mapsto\Set,P\mapsto L,Q\mapsto F,\mu\mapsto\mu^L,\eta\mapsto\eta^L,\\
\phantom{ddddd}\mu'\mapsto\mu^F,\eta'\mapsto\eta^F,\otimes\mapsto\otimes^L\}$,
\item
$\Psi = \{O\mapsto\Set,P\mapsto M,Q\mapsto F,\mu\mapsto\mu^M,\eta\mapsto\eta^M,\\
\phantom{ddddd}\mu'\mapsto\mu^F,\eta'\mapsto\eta^F,\otimes\mapsto\otimes^M\}$,
\item
$(\sharp)_{P} = \sharp$, $(\sharp)_{Q} = \id_F$.
\end{itemize}

We need to verify the conditions of Theorem \ref{thm:nattranssharp} and Corollary \ref{cor:nattranssharp}. For the condition \eqref{dia:nattranssharp} of Theorem \ref{thm:nattranssharp}, we must show that \eqref{dia:nattranssharp} commutes for each of the 2-cell generators in $\RR$. These conditions are
\begin{align*}
&
\begin{tikzpicture}[->, >=stealth', auto, node distance=18mm]
\scriptsize
\node (NW) {$LL$};
\node (NE) [right of=NW] {$L$};
\node (SW) [below of=NW, node distance=10mm] {$MM$};
\node (SE) [below of=NE, node distance=10mm] {$M$};
\path (NW) edge node {$\mu^L$} (NE);
\path (NW) edge node[swap] {$\sharp\sharp$} (SW);
\path (NE) edge node {$\sharp$} (SE);
\path (SW) edge node[swap] {$\mu^M$} (SE);
\end{tikzpicture}
&&
\begin{tikzpicture}[->, >=stealth', auto, node distance=18mm]
\scriptsize
\node (NW) {$I$};
\node (NE) [right of=NW] {$L$};
\node (SW) [below of=NW, node distance=10mm] {$I$};
\node (SE) [below of=NE, node distance=10mm] {$M$};
\path (NW) edge node {$\eta^L$} (NE);
\path (NW) edge node[swap] {$\id$} (SW);
\path (NE) edge node {$\sharp$} (SE);
\path (SW) edge node[swap] {$\eta^M$} (SE);
\end{tikzpicture}
&&
\begin{tikzpicture}[->, >=stealth', auto, node distance=18mm]
\scriptsize
\node (NW) {$LF$};
\node (NE) [right of=NW] {$FL$};
\node (SW) [below of=NW, node distance=10mm] {$MF$};
\node (SE) [below of=NE, node distance=10mm] {$FM$};
\path (NW) edge node {$\otimes^L$} (NE);
\path (NW) edge node[swap] {$\sharp F$} (SW);
\path (NE) edge node {$F\sharp$} (SE);
\path (SW) edge node[swap] {$\otimes^M$} (SE);
\end{tikzpicture}
\end{align*}
plus two more conditions for $\mu^F$ and $\eta^F$, which only involve identities and are trivial.
The rightmost diagram is just \eqref{eq:otimesdef}, which was verified in Lemma \ref{lem:otimesdef}. A complete proof for the remaining two is given in \S\ref{sec:LLMM}.

For Corollary \ref{cor:nattranssharp}, we must show that all $(\sharp)_S$ are epimorphisms. In light of \eqref{eq:horizC}, we need only show that $\sharp$ is an epimorphism, and that the functors $L$, $F$, and $M$ preserve epimorphisms. For any set $X$, $\sharp_X:LX\to MX$ takes a list of elements of $X$ to the multiset of its elements and is clearly surjective. That the functors $L$, $F$, and $M$ preserve epimorphisms is evident after a few moments' thought; but we provide an explicit proof in \S\ref{sec:LLMM}.

These results lead to an alternative proof of Theorem \ref{thm:distribM} using a general method to transfer properties involving $L$ to properties involving $M$. Since our application satisfies the conditions of Theorem \ref{thm:nattranssharp} and Corollary \ref{cor:nattranssharp}, we have that any 2-cell equation that holds under the interpretation $\Phi$ also holds under $\Psi$. This allows us to derive several needed equations involving $M$ all at once from the corresponding equations involving $L$: the monad equations for $\mu^M:MM\to M$ and $\eta^M:I\to M$, the naturality of $\otimes^M$, and the equations for the distributive law $\otimes^M:MF\to FM$.

To see how this works, observe that the diagrams for these properties (Figures \ref{fig:nattrans} and \ref{fig:distrib} for example) all have the same character: they are all cubical diagrams consisting of a diagram involving $L$ on the top face, the same diagram on the bottom face with $L$ replaced by $M$, and vertical arrows connecting them consisting of components of the 2-natural transformation $(\sharp)$. These components just replace $L$'s with $M$'s and are epimorphisms. The vertical faces all commute, because $(\sharp)$ is a 2-natural transformation. One can then conclude that the bottom face commutes whenever the top face does by the same argument as in Lemma \ref{lem:otimesdef}.

\section{Conclusion}

We have provided a simplified proof of a distributive law between the finite distribution and finite multiset monads, first proved by Jacobs \cite{Jacobs21}. Unlike Jacobs's proof, we do not require any combinatorial machinery, which we view as not germane. The proof consists of two parts: (a) a distributive law between the distribution and list monads, which is more intuitive and easier to prove equationally due to the ordering of lists, and (b) a general theorem on 2-categories to the effect that category-theoretic laws, such as monad and distributive laws, are preserved by epic 2-natural transformations. In this case, laws involving lists are mapped to laws involving multisets via the Parikh map that forgets order, taking a list to its multiset of elements. 

We believe that Jacobs's result has the potential to significantly advance the understanding of the interaction of nondeterminism and probability, especially in light of previous well-known difficulties combining the distribution and powerset monads. We did not review Jacobs's proof in this paper, as it is heavily technical and not related to our approach in any meaningful way. However, we do believe that important results deserve more than one proof, and that some readers may find our approach a more accessible alternative.

An interesting and novel aspect of our work is the general 2-categorical formulation of the problem, in which 2-categories are presented as a logical vehicle for the equational logic of natural transformations. There are several innovations here that we believe deserve further study and that may find further application: the use of free 2-categories as logical syntax, the use of 2-functors as algebraic interpretations, and the observation that epic 2-natural transformations preserve commutative diagrams involving natural transformations. We note explicitly the elegant way the proof plays out with both horizontal and vertical composition, leading to the preservation of an entire class of equations between natural transformations as special cases. We suspect that there may be potential for further applications of this technique involving other classes of coherence diagrams.

For simplicity, we have derived the distributive law for the finite multiset monad and finite distribution monad only; however, we are confident that the proof generalizes to the full Giry monad of measures on an arbitrary measurable spaces, as claimed in \cite{Dash21}, using the same technique.

\section*{Acknowledgments}

The support of the National Science Foundation under grant CCF-2008083 is gratefully acknowledged.

\bigskip




\appendix

\section{Extra proofs}
\subsection{$\otimes^L:LF\to FL$ is a distributive law}
\label{sec:distribL}

In this section we verify that $\otimes^L:LF\to FL$ is a distributive law. Because we are working with lists instead of multisets, the task is notationally much less cumbersome and involves only elementary equational reasoning. This is one major benefit of our approach.

\begin{proof}[Proof of Theorem \ref{thm:LF}]
We first argue that the diagram
\begin{align*}
\begin{tikzpicture}[->, >=stealth', auto, node distance=20mm]
\small
\node (NW) {$LFX$};
\node (NE) [right of=NW] {$FLX$};
\node (SW) [below of=NW, node distance=12mm] {$LFY$};
\node (SE) [below of=NE, node distance=12mm] {$FLY$};
\path (NW) edge node {$\otimes^L_X$} (NE);
\path (NW) edge node[swap] {$LFh$} (SW);
\path (NE) edge node {$FLh$} (SE);
\path (SW) edge node[swap] {$\otimes^L_Y$} (SE);
\end{tikzpicture}
\end{align*}
commutes, thus $\otimes^L:LF\to FL$ is a natural transformation.
 
Let $\bar f = \bx{\seq f1n}\in LFX$, $f_i\in FX$. By definition, for $\bar a=\bx{\seq a1n}\in X^n$,
\begin{align}
& \otimes^L_X(\bar f)(\bar a) = \prod_{i=1}^n f_i(a_i)\label{eq:otimesnat0}
\end{align}
and $\otimes^L_X(\bar f)(\bar a) = 0$ if the lists $\bar f$ and $\bar a$ are not the same length.
The support of $\otimes^L_X(\bar f)$ is $\supp\otimes^L_X(\bar f) = \prod_{i=1}^n\supp f_i$.

If $h:X\to Y$, then $Fh:FX\to FY$.
If $f\in FX$ and $b\in Y$, then $Fh(f)\in FY$ and
\begin{align}
& Fh(f) = \lam{b\in Y}{\sum_{\substack{a\in X\\h(a)=b}} f(a)}\nonumber\\
& Fh(f)(b) = \sum_{\substack{a\in X\\h(a)=b}} f(a).\label{eq:otimesnat1}
\end{align}
Although there may be infinitely many $a\in X$ with $h(a)=b$, since $f$ is of finite support, at most finitely many of these contribute to the sum.

Using \eqref{eq:otimesnat1} and \eqref{eq:otimesnat0},
\begin{align}
FLh(\otimes^L_X(\bar f))(\bar b)
&= \sum_{\substack{\bar a\in X^n\\Lh(\bar a)=\bar b}} \otimes^L_X(\bar f)(\bar a)
= \sum_{\substack{\bar a\in X^n\\Lh(\bar a)=\bar b}} \prod_{i=1}^n f_i(a_i).\label{eq:otimesnat5}
\end{align}
On the other hand, using \eqref{eq:otimesnat0} and \eqref{eq:otimesnat1},
\begin{align}
\otimes^L_Y(LFh(\bar f))(\bar b)
&= \prod_{i=1}^n Fh(f_i)(b_i)
= \prod_{i=1}^n \sum_{\substack{a\in X\\h(a)=b_i}} f_i(a).\label{eq:otimesnat6}
\end{align}
The right-hand sides of \eqref{eq:otimesnat5} and \eqref{eq:otimesnat6} are equal by the distributive law of semirings, thus
\begin{align*}
\otimes^L_Y(LFh(\bar f))(\bar b)
&= FLh(\otimes^L_X(\bar f))(\bar b).
\end{align*}
As $\bar f$ and $\bar b$ were arbitrary,
\begin{align*}
\otimes^L_Y\circ LFh &= FLh\circ\otimes^L_X.
\end{align*}

We must show that the diagrams of Fig.~\ref{fig:distribL} commute.
Let us start with the easy ones. For $f\in FX$,
\begin{align*}
& \otimes^L_X(\eta^L_{FX}(f))\\
&= \otimes^L_X(\bx{f}) = \lam{\bx{a}}{f(a)}\\
&= \lam y{\sum_{\eta^L_X(a)=y} f(a)} = F\eta^L_X(f).
\end{align*}
For $\bar a = \bx{\seq a1n}\in LX$,
\begin{align*}
& \otimes^L_X(L\eta^F_X(\bar a))\\
&= \otimes^L_X(\bx{\eta^F_X(a_1),\ldots,\eta^F_X(a_n)})\\
&= \otimes^L_X(\bx{\lam b{[b=a_1]},\ldots,\lam b{[b=a_n]}})\\
&= \lam{\bx{\seq b1n}\in X^n}{\prod_{i=1}^n (\lam b{[b=a_i]})(b_i)}\\
&= \lam{\bx{\seq b1n}\in X^n}{\prod_{i=1}^n {[b_i=a_i]}}\\
&= \lam{\bar b\in X^n}{[\bar b=\bar a]}\\
&= \eta^F_{LX}(\bar a).
\end{align*}

For the top-left diagram of Fig.~\ref{fig:distribL}, we need to show
\begin{align*}
F\mu^L_X\circ \otimes^L_{LX}\circ L\otimes^L_X = \otimes^L_X\circ \mu^L_{FX}.
\end{align*}
Let $\bx{\seq{\bar f}1n}\in LLFX$, where $\bar f_i = \bx{f_{i1},\ldots,f_{ik_i}}$, $1\le i\le n$. Then
\begin{align*}
& F\mu^L_X(\otimes^L_{LX}(L\otimes^L_X(\bx{\seq{\bar f}1n})))\\
&= F\mu^L_X(\otimes^L_{LX}(\bx{\otimes^L_X(\bar f_1),\ldots,\otimes^L_X(\bar f_1)}))\\
&= F\mu^L_X(\otimes^L_{LX}(\bx{\lam{\bar a_1}{\prod_{i=1}^{k_1}f_{1i}(a_{1i})},\ldots,\lam{\bar a_n}{\prod_{i=1}^{k_n}f_{1i}(a_{1i})}})))\\
&= F\mu^L_X(\lam{\bx{\bar a_1,\ldots,\bar a_n}}{\prod_{i=1}^n\prod_{j=1}^{k_i}f_{ij}(a_{ij})})\\
&= \lam y{\sum_{\mu^L_X(x)=y} (\lam{\bx{\bar a_1,\ldots,\bar a_n}}{\prod_{i=1}^n\prod_{j=1}^{k_i}f_{ij}(a_{ij})})(x)}\\
&= \lam{\bx{a_{11},\ldots,a_{1k_1},\ldots,a_{n1},\ldots,a_{nk_n}}}{\prod_{i=1}^n\prod_{j=1}^{k_i}f_{ij}(a_{ij})}\\
&= \otimes^L_X(\bx{f_{11},\ldots,f_{1k_1},\ldots,f_{n1},\ldots,f_{nk_n}})\\
&= \otimes^L_X(\mu^L_{FX}(\bx{\seq{\bar f}1n})).
\end{align*}

Finally, for the bottom-left diagram of Fig.~\ref{fig:distribL}, we need to show
\begin{align*}
\mu^F_{LX}\circ F\otimes^L_X\circ \otimes^L_{FX} = \otimes^L_X\circ L\mu^F_X.
\end{align*}
Let $\bar f = \bx{f_1,\ldots,f_n}\in LFFX$, $f_i\in FFX = \reals^{\reals^X}$. One path in the diagram gives
\begin{align}
& \mu^F_{LX}(F\otimes^L_X(\otimes^L_{FX}(\bar f)))\nonumber\\
&= \mu^F_{LX}(F\otimes^L_X(\lam{\bar g}{\prod_{i=1}^n f_i(g_i)}))\nonumber\\
&= \mu^F_{LX}(\lam y{\sum_{\otimes^L_X(x)=y} (\lam{\bar g}{\prod_{i=1}^n f_i(g_i)})(x)})\nonumber\\
&= \mu^F_{LX}(\lam y{\sum_{\otimes^L_X(\bar g)=y} {\prod_{i=1}^n f_i(g_i)}})\nonumber\\
&= \lam{a\in LX}{\sum_{g\in FLX}{(\lam y{\sum_{\otimes^L_X(\bar g)=y} {\prod_{i=1}^n f_i(g_i)}})(g)\cdot g(a)}}\nonumber\\
&= \lam{a\in LX}{\sum_{g\in FLX}{({\sum_{\otimes^L_X(\bar g)=g} {\prod_{i=1}^n f_i(g_i)}})\cdot g(a)}}\nonumber\\
&= \lam{a\in LX}{\sum_{\bar g\in (FX)^n}{{{\prod_{i=1}^n f_i(g_i)}}\cdot (\otimes^L_X(\bar g))(a)}}\nonumber\\
&= \lam{\bx{\seq a1n}\in LX}{\sum_{\bar g\in (FX)^n}{{{\prod_{i=1}^n f_i(g_i)}}\cdot \prod_{i=1}^n g_i(a_i)}}\nonumber\\
&= \lam{\bx{\seq a1n}}{\sum_{\bar g\in (FX)^n}\prod_{i=1}^n{f_i(g_i)\cdot g_i(a_i)}}.\label{eq:zzdr1}
\end{align}
The other path gives
\begin{align}
& \otimes^L_X(L\mu^F_X(\bar f))\nonumber\\
&= \otimes^L_X(\bx{\mu^F_X(f_1),\ldots,\mu^F_X(f_n)})\nonumber\\
&= \otimes^L_X(\bx{\lam{a\in X}{\sum_{g\in FX}{f_1(g)\cdot g(a)}},\ldots,\nonumber\\
& \phantom{eeeeeeee.}\lam{a\in X}{\sum_{g\in FX}{f_n(g)\cdot g(a)}}}))\nonumber\\
&= \lam{\bx{\seq a1n}}{\prod_{i=1}^n \sum_{g\in FX}{f_i(g)\cdot g(a_i)}},\label{eq:zzdr2}
\end{align}
and \eqref{eq:zzdr1} and \eqref{eq:zzdr2} are equivalent by distributivity of multiplication over addition.
\end{proof}

\subsection{Conditions of Theorem \ref{thm:nattranssharp}}
\label{sec:LLMM}

We must also establish the remaining conditions \eqref{dia:nattranssharp} required for the application of Theorem \ref{thm:nattranssharp}. Again, the proofs involve only elementary equational reasoning.
\begin{theorem}
\label{thm:LLMM}
The following diagrams commute:
\begin{align*}
&
\begin{tikzpicture}[->, >=stealth', auto, node distance=20mm]
\small
\node (NW) {$LL$};
\node (NE) [right of=NW] {$L$};
\node (SW) [below of=NW, node distance=12mm] {$MM$};
\node (SE) [below of=NE, node distance=12mm] {$M$};
\path (NW) edge node {$\mu^L$} (NE);
\path (NW) edge node[swap] {$\sharp\sharp$} (SW);
\path (NE) edge node {$\sharp$} (SE);
\path (SW) edge node[swap] {$\mu^M$} (SE);
\end{tikzpicture}
&&
\begin{tikzpicture}[->, >=stealth', auto, node distance=20mm]
\small
\node (NW) {$I$};
\node (NE) [right of=NW] {$L$};
\node (SW) [below of=NW, node distance=12mm] {$I$};
\node (SE) [below of=NE, node distance=12mm] {$M$};
\path (NW) edge node {$\eta^L$} (NE);
\path (NW) edge node[swap] {$\id$} (SW);
\path (NE) edge node {$\sharp$} (SE);
\path (SW) edge node[swap] {$\eta^M$} (SE);
\end{tikzpicture}
\end{align*}
\end{theorem}
\begin{proof}
For the left-hand diagram, let $\ell_i = \bx{a_{i1},\ldots,a_{ik_i}}$, $a_{ij}\in X$, $1\le i\le n$, and let
\begin{align*}
f_i
&= \sharp_X(\ell_i) = \sum_{j=1}^{k_i}\sharp_X(\bx{a_{ij}})\\
&= \sum_{j=1}^{k_i}\lam{a\in X}{[a=a_{ij}]},\ \ 1\le i\le n.
\end{align*}
Recall that $\sharp\sharp = \sharp M\circ L\sharp$. Then
\begin{align*}
& \sharp_X(\mu^L_X(\bx{\ell_1,\ldots,\ell_n}))\\
&= \sharp_X(\bx{a_{11},\ldots,a_{1k_1},\ldots,a_{n1},\ldots,a_{nk_n}})\\
&= \sum_{i=1}^n\sum_{j=1}^{k_i} \sharp_X(\bx{a_{ij}})
= \sum_{i=1}^n f_i,
\end{align*}
\begin{align*}
&\mu^M_X(\sharp_{MX}(L\sharp_X(\bx{\seq\ell 1n})))\\
&= \mu^M_X(\sharp_{MX}(\bx{\sharp_X(\ell_1),\ldots,\sharp_X(\ell_n)}))\\
&= \mu^M_X(\sharp_{MX}(\bx{f_1,\ldots,f_n}))
= \mu^M_X(\sum_{i=1}^n \sharp_{MX}(\bx{f_i}))\\
&= \mu^M_X(\sum_{i=1}^n \lam{h\in FX}{[h=f_i]})\\
&= \lam{a\in X}{\sum_{h\in FX} (\sum_{i=1}^n \lam{h\in FX}{[h=f_i]})(h)\cdot h(a)}\\
&= \lam{a\in X}{\sum_{h\in FX} \sum_{i=1}^n [h=f_i]\cdot h(a)}\\
&= \lam{a\in X}{\sum_{i=1}^n f_i(a)}
= \sum_{i=1}^n f_i.
\end{align*}
Finally, for the right-hand diagram, for $a\in X$,
\begin{align*}
\sharp_X(\eta^L_X(a))
&= \sharp_X(\bx a)
= \lam{b\in X}{[b=a]}
= \eta^M(a).
\tag*{\qed}
\end{align*}
\end{proof}

\begin{lemma}
The functors $L$, $F$, and $M$ preserve epimorphisms.
\end{lemma}
\begin{proof}
Actually, all functors defined on \Set\ preserve epimorphisms. By the axiom of choice, every surjection $f:X\to Y$ has a right inverse $f':Y\to X$, that is, $f\circ f'=\id_Y$. This says that $f$ is a split epimorphism, and all functors $G$ in any category preserve split epimorphisms:
\begin{align*}
Gf\circ Gf' &= G(f\circ f') = G\id_Y = \id_{GY}.
\end{align*}
A split epimorphism $f$ in any category is an epimorphism, since
\begin{align*}
e_1\circ f = e_2\circ f &\Imp e_1\circ f\circ f' = e_2\circ f\circ f' \Imp e_1 = e_2.
\tag*{\qed}
\end{align*}
\end{proof}

\end{document}